\documentclass[copyright,creativecommons]{eptcs}
\providecommand{\event}{EXPRESS/SOS 2015} 

\usepackage{color}

\usepackage[utf8]{inputenc}
\usepackage{subfigure}
\usepackage[english]{babel}
\usepackage{amsmath,amsfonts,amssymb}
\usepackage{proof}
\usepackage{graphicx}
\usepackage{multicol}
\usepackage{listings}
\usepackage{url}
\usepackage{cmll}
\usepackage{algpseudocode}
\usepackage{algorithmicx}
\usepackage{enumerate}
\usepackage{textcomp}
\usepackage{latexsym}
\usepackage{epsfig} 
\usepackage{cite} 
\usepackage{xcolor}
\usepackage{bbm}
\usepackage{times}

\usepackage{enumerate}
\usepackage{hyperref}
\usepackage{stmaryrd}
\usepackage{xspace}
\usepackage{amsthm}

\usepackage{packs/mathpartir}
\usepackage{packs/proof-dashed}

\usepackage{makeidx}

\newif{\ifSHORT}
\newif{\iflong}
\newif{\ifLongVersion}
\newif{\ifWithRecords}
\newif{\ifWithProofs}

\newtheorem{theorem}{Theorem}[section]
\newtheorem{lemma}[theorem]{Lemma}

\newtheorem{proposition}[theorem]{Proposition}

\newtheorem{notation}[theorem]{Notation}
\newtheorem{corollary}[theorem]{Corollary}
\newtheorem{definition}[theorem]{Definition}

\newcommand{\Koba}[1]{\ensuremath{\mathcal{K}_{#1}}\xspace}
\newcommand{\fullKoba}{\ensuremath{\mathcal{K}\xspace}}
\newcommand{\lkoba}{\Koba{n}}
\newcommand{\lpkoba}{\Koba{1}}

\newcommand{\lcp}{\ensuremath{\mathcal{L}}\xspace}

\newcommand{\m}[1]{\mathsf{#1}}

\newcommand{\wt}[1]{\widetilde{#1}}
\newcommand{\tl}[1]{\tilde{#1}}
\newcommand{\defeq}{\triangleq}

\newcommand{\emp}{\emptyset}

\newcommand{\inl}{\mathtt{inl}}
\newcommand{\inr}{\mathtt{inr}}
\newcommand{\inx}{\mathtt{inx}}

\def\midd{\; \; \mbox{\Large{$\mid$}}\;\;}

\newcommand{\p}{$\pi$-\! }
\newcommand{\nil}{{\mathbf{0}}}
\newcommand{\selection}[2]{#1\triangleleft {#2}}

\newcommand{\branching}[3]{#1\triangleright\{{#2}_i:#3_i\}_{i\in I}}
\newcommand{\parbranching}[3]{#1\triangleright\{{#2}:#3\}_{i\in I}}

\newcommand{\out}[1]{\langle #1\rangle}
\newcommand{\bout}[2]{\overline{#1}(#2)}
\newcommand{\inp}[1]{(#1)}
\newcommand{\res}[1]{({\boldsymbol \nu} #1)}
\newcommand{\pp}{\ \boldsymbol{|}\ }
\newcommand{\picase}[3]
{\mathbf{case} \, {#1} \, \mathbf{of}\, \{ \mathnormal{l}_i \_ {#2} \triangleright {#3} \}_{i\in I}}

\newcommand{\vv}[2]{\mathnormal{l}_{#1}\_ #2}

\newcommand{\uptok}{\doteqdot}

\newcommand{\variant}[2]{\langle {#1} : #2 \rangle_{i\in I}}

\newcommand{\select}[2]{\oplus\{{#1}_i:#2_i\}_{i\in I}}
\newcommand{\branch}[2]{\&\{{#1}_i:#2_i\}_{i\in I}}

\newcommand{\nilT}{\ensuremath{{\bf end}}}

\newcommand{\unit}{{\mathbf{n}}}

\newcommand{\ct}{\mathsf{c}}

\newcommand{\su}{\mathsf{su}}
\newcommand{\inpuse}[2]{\wn^{#1}_{#2}}
\newcommand{\outuse}[2]{\oc^{#1}_{#2}}
\newcommand{\ca}{\kappa}
\newcommand{\ob}{\omicron}
\newcommand{\obs}[2]{\mathsf{ob}_{#1}(#2)}
\newcommand{\caps}[2]{\mathsf{cap}_{#1}(#2)}
\newcommand{\conpar}[2]{\mathsf{con}_{#1}(#2)}
\newcommand{\rel}[1]{\mathsf{rel}( #1 )}
\newcommand{\con}[1]{\mathsf{con}(#1)}

\newcommand{\Didtext}[1]{{\footnotesize{\textsc{#1}}}}
\newcommand{\Did}[1]{({\Didtext{#1}})}
\newcommand{\lf}{\vdash^{n}_{\mathtt{KB}}}
\newcommand{\plf}{\vdash^{1}_{\mathtt{KB}}}
\newcommand{\s}{\vdash_{\mathtt{{ST}}}}
\newcommand{\cp}{\vdash_{\mathtt{CH}}}

\newcommand{\fn}[1]{\mathsf{fn}(#1)}   

\newcommand{\catal}[2]{\mathcal{C}_{{#1}}\big[{#2}\big]}

\newcommand{\semi}{\ {\bf ; }\ }
\newcommand{\mto}{\to^*}

\newcommand{\dom}{\m{dom}}

\newcommand{\chrp}[2]{\ensuremath{\{\!\!| #1 |\!\!\}^{#2}}}
\newcommand{\chr}[1]{\ensuremath{\{\!\!\!\{ #1 \}\!\!\!\}}}
\newcommand{\enco}[2]{\llbracket #1 \rrbracket_{#2}}
\newcommand{\encf}[1]{\llbracket #1 \rrbracket_{\mathnormal{f}}}
\newcommand{\encCP}[1]{\llparenthesis #1 \rrparenthesis}
\newcommand{\encoCP}[2]{\llparenthesis #1 \rrparenthesis^{#2}}
\newcommand{\f}[1]{{\mathnormal f}_{#1}}

\newcommand{\enc}[2]{\llbracket #1 \rrbracket_{\mathnormal{f}, \{x,y \mapsto {#2}\}}}
\newcommand{\encx}[2]{\llbracket #1 \rrbracket_{\mathnormal{f},\{x\mapsto {#2}\}}}

\newcommand{\fakepar}[1]{\ensuremath{\parallel_#1}}
 
\newcommand{\lolli}{\mathord{\multimap}}

\newcommand{\parl}{\mathbin{\bindnasrepma}}
\newcommand{\one}{{\bf 1}}
\newcommand{\CLL}{\texttt{CLL}\xspace}
\newcommand{\dual}[1]{\overline{#1}}

\newcommand{\nub}{{\boldsymbol{\nu}}}
\newcommand{\linkr}[2]{[#1\!\leftrightarrow\!#2]}

\newcommand{\cut}{\mathsf{cut}}
\newcommand{\mix}{\mathsf{mix}}

\newcommand{\tid}{\mathsf{id}}
\newcommand{\ov}[1]{\overline{#1}}
\newcommand{\para}{\mathord{\;\boldsymbol{|}\;}}
\def\substj#1#2{[\raisebox{.5ex}{\small$#1$}\! / \mbox{\small$#2$}]}

\newcommand{\D}{\Delta}

\newcommand{\tra}[1]{\xrightarrow{#1}}
\newcommand{\redd}{\tra{~~~}}

\newdimen\proofrulebreadth \proofrulebreadth=.05em
\newdimen\proofdotseparation \proofdotseparation=1.25ex
\newdimen\proofrulebaseline \proofrulebaseline=2ex
\newcount\proofdotnumber \proofdotnumber=3
\let\then\relax
\def\hfi{\hskip0pt plus.0001fil}
\mathchardef\squigto="3A3B
\newif\ifinsideprooftree\insideprooftreefalse
\newif\ifonleftofproofrule\onleftofproofrulefalse
\newif\ifproofdots\proofdotsfalse
\newif\ifdoubleproof\doubleprooffalse
\let\wereinproofbit\relax
\newdimen\shortenproofleft
\newdimen\shortenproofright
\newdimen\proofbelowshift
\newbox\proofabove
\newbox\proofbelow
\newbox\proofrulename
\def\shiftproofbelow{\let\next\relax\afterassignment\setshiftproofbelow\dimen0 }
\def\shiftproofbelowneg{\def\next{\multiply\dimen0 by-1 }%
\afterassignment\setshiftproofbelow\dimen0 }
\def\setshiftproofbelow{\next\proofbelowshift=\dimen0 }
\def\setproofrulebreadth{\proofrulebreadth}

\def\prooftree{
\ifnum  \lastpenalty=1
\then   \unpenalty
\else   \onleftofproofrulefalse
\fi
\ifonleftofproofrule
\else   \ifinsideprooftree
        \then   \hskip.5em plus1fil
        \fi
\fi
\bgroup
\setbox\proofbelow=\hbox{}\setbox\proofrulename=\hbox{}%
\let\justifies\proofover\let\leadsto\proofoverdots\let\Justifies\proofoverdbl
\let\using\proofusing\let\[\prooftree
\ifinsideprooftree\let\]\endprooftree\fi
\proofdotsfalse\doubleprooffalse
\let\thickness\setproofrulebreadth
\let\shiftright\shiftproofbelow \let\shift\shiftproofbelow
\let\shiftleft\shiftproofbelowneg
\let\ifwasinsideprooftree\ifinsideprooftree
\insideprooftreetrue
\setbox\proofabove=\hbox\bgroup$\displaystyle 
\let\wereinproofbit\prooftree
\shortenproofleft=0pt \shortenproofright=0pt \proofbelowshift=0pt
\onleftofproofruletrue\penalty1
}

\def\eproofbit{
\ifx    \wereinproofbit\prooftree
\then   \ifcase \lastpenalty
        \then   \shortenproofright=0pt  
        \or     \unpenalty\hfil         
        \or     \unpenalty\unskip       
        \else   \shortenproofright=0pt  
        \fi
\fi
\global\dimen0=\shortenproofleft
\global\dimen1=\shortenproofright
\global\dimen2=\proofrulebreadth
\global\dimen3=\proofbelowshift
\global\dimen4=\proofdotseparation
\global\count255=\proofdotnumber
$\egroup  
\shortenproofleft=\dimen0
\shortenproofright=\dimen1
\proofrulebreadth=\dimen2
\proofbelowshift=\dimen3
\proofdotseparation=\dimen4
\proofdotnumber=\count255
}

\def\proofover{
\eproofbit 
\setbox\proofbelow=\hbox\bgroup 
\let\wereinproofbit\proofover
$\displaystyle
}%
\def\proofoverdbl{
\eproofbit 
\doubleprooftrue
\setbox\proofbelow=\hbox\bgroup 
\let\wereinproofbit\proofoverdbl
$\displaystyle
}%
\def\proofoverdots{
\eproofbit 
\proofdotstrue
\setbox\proofbelow=\hbox\bgroup 
\let\wereinproofbit\proofoverdots
$\displaystyle
}%
\def\proofusing{
\eproofbit 
\setbox\proofrulename=\hbox\bgroup 
\let\wereinproofbit\proofusing
\kern0.3em$
}

\def\endprooftree{
\eproofbit 
  \dimen5 =0pt
\dimen0=\wd\proofabove \advance\dimen0-\shortenproofleft
\advance\dimen0-\shortenproofright
\dimen1=.5\dimen0 \advance\dimen1-.5\wd\proofbelow
\dimen4=\dimen1
\advance\dimen1\proofbelowshift \advance\dimen4-\proofbelowshift
\ifdim  \dimen1<0pt
\then   \advance\shortenproofleft\dimen1
        \advance\dimen0-\dimen1
        \dimen1=0pt
        \ifdim  \shortenproofleft<0pt
        \then   \setbox\proofabove=\hbox{%
                        \kern-\shortenproofleft\unhbox\proofabove}%
                \shortenproofleft=0pt
        \fi
\fi
\ifdim  \dimen4<0pt
\then   \advance\shortenproofright\dimen4
        \advance\dimen0-\dimen4
        \dimen4=0pt
\fi
\ifdim  \shortenproofright<\wd\proofrulename
\then   \shortenproofright=\wd\proofrulename
\fi
\dimen2=\shortenproofleft \advance\dimen2 by\dimen1
\dimen3=\shortenproofright\advance\dimen3 by\dimen4
\ifproofdots
\then
        \dimen6=\shortenproofleft \advance\dimen6 .5\dimen0
        \setbox1=\vbox to\proofdotseparation{\vss\hbox{$\cdot$}\vss}%
        \setbox0=\hbox{%
                \advance\dimen6-.5\wd1
                \kern\dimen6
                $\vcenter to\proofdotnumber\proofdotseparation
                        {\leaders\box1\vfill}$%
                \unhbox\proofrulename}%
\else   \dimen6=\fontdimen22\the\textfont2 
        \dimen7=\dimen6
        \advance\dimen6by.5\proofrulebreadth
        \advance\dimen7by-.5\proofrulebreadth
        \setbox0=\hbox{%
                \kern\shortenproofleft
                \ifdoubleproof
                \then   \hbox to\dimen0{%
                        $\mathsurround0pt\mathord=\mkern-6mu%
                        \cleaders\hbox{$\mkern-2mu=\mkern-2mu$}\hfill
                        \mkern-6mu\mathord=$}%
                \else   \vrule height\dimen6 depth-\dimen7 width\dimen0
                \fi
                \unhbox\proofrulename}%
        \ht0=\dimen6 \dp0=-\dimen7
\fi
\let\doll\relax
\ifwasinsideprooftree
\then   \let\VBOX\vbox
\else   \ifmmode\else$\let\doll=$\fi
        \let\VBOX\vcenter
\fi
\VBOX   {\baselineskip\proofrulebaseline \lineskip.2ex
        \expandafter\lineskiplimit\ifproofdots0ex\else-0.6ex\fi
        \hbox   spread\dimen5   {\hfi\unhbox\proofabove\hfi}%
        \hbox{\box0}%
        \hbox   {\kern\dimen2 \box\proofbelow}}\doll%
\global\dimen2=\dimen2
\global\dimen3=\dimen3
\egroup 
\ifonleftofproofrule
\then   \shortenproofleft=\dimen2
\fi
\shortenproofright=\dimen3
\onleftofproofrulefalse
\ifinsideprooftree
\then   \hskip.5em plus 1fil \penalty2
\fi
}

\title{Comparing Deadlock-Free Session Typed Processes}
\author{Ornela Dardha
\institute{University of Glasgow, United Kingdom}
\and
Jorge A. P\'{e}rez
\institute{University of Groningen, The Netherlands}
}
\def\titlerunning{Comparing Deadlock-Free Session Typed Processes}
\def\authorrunning{O. Dardha \& J.A. P\'{e}rez}
\begin{document}
\maketitle

\begin{abstract}
Besides 
respecting prescribed 
protocols, communication-centric systems should never ``get stuck''. 
This 
requirement has been expressed by liveness properties such as progress or (dead)lock freedom. Several typing disciplines that ensure these properties for mobile processes
have been proposed. Unfortunately, very little is known about 
the precise relationship between these disciplines--and the classes of typed processes they induce. 

In this paper, we compare \lcp and \fullKoba,
two classes of deadlock-free, session typed concurrent processes.
The class \lcp 
stands out for its canonicity: it 
results naturally from interpretations of linear logic propositions as session types.
The class \fullKoba, obtained by encoding session types into Kobayashi's usage types, 
includes processes not typable in other type systems. 

We show that \lcp is strictly included in \fullKoba.
We also identify the precise condition under which \lcp and \fullKoba coincide. One key observation is that the \emph{degree of sharing} between parallel processes 
determines a new expressiveness hierarchy for  typed processes.
We also provide a type-preserving rewriting procedure of processes in \fullKoba into processes in \lcp. 
This procedure suggests that, while effective, the degree of sharing is a rather subtle  criteria for distinguishing typed processes.
\end{abstract}
%


\section{Introduction}
The goal of this work is to 
formally relate
different type systems for the $\pi$-calculus.
Our interest is in \emph{session-based concurrency}, a type-based approach to communication correctness: dialogues between participants are structured into \emph{sessions},  basic   communication units; descriptions of interaction sequences are then abstracted as \emph{session types}~\cite{HVK98} which are checked against process specifications. 
We offer the first 
formal comparison between  different {type systems} that enforce \emph{(dead)lock freedom}, 
the liveness property 
that ensures session communications never ``get stuck''.
Our approach relates the classes of typed processes that such systems induce. 
To this end, 
we identify a property on the structure of typed parallel processes, the \emph{degree of sharing}, which 
is key in distinguishing two salient classes of deadlock-free session processes,
and in shedding light on their formal underpinnings.

In session-based concurrency,  
types enforce correct communications through different safety and liveness properties. 
Basic correctness properties are \emph{communication safety} and \emph{session fidelity}: 
while the former ensures absence of errors (e.g., communication mismatches), the latter 
ensures that well-typed processes respect the  protocols  prescribed  by session types.
Moreover,
a central (liveness) property for safe processes is that they should never ``get stuck''.
This is the well-known \emph{progress} property, which 
asserts that a well-typed term either is a final value or can further reduce~\cite{Pierce02}.
In calculi for concurrency, this 
 property has been formalized as
\emph{deadlock freedom}
(``a process is deadlock-free if it can always reduce until it eventually terminates, unless the whole process diverges''~\cite{K06}) or as
\emph{lock freedom} (``a process is lock free if it can always reduce until it eventually terminates, even if the whole process diverges''~\cite{K02}).
Notice that in the absence of divergent behaviors, deadlock and lock freedom coincide.

(Dead)lock freedom
guarantees that all communications will eventually succeed, 
an appealing requirement for communicating processes.
Several advanced type disciplines that ensure deadlock-free  processes have been proposed~(see, e.g.,\cite{K02,K06,DLY07,CD10,CairesP10,CairesPT12,P13,DBLP:conf/coordination/VieiraV13}).
Unfortunately, these disciplines 
consider different process languages and/or 
are based on rather different principles.
As a result, very little is known about how they relate to each other.
This begs several research questions: What is the formal relationship between these type disciplines?
What  classes of deadlock-free processes do they induce?
%
%

In this paper, we tackle these open questions by 
comparing
\lcp and \fullKoba,
two salient classes of deadlock-free, session typed processes (Definition~\ref{d:lang}):
\begin{enumerate}[$\bullet$]
\item $\lcp$ contains all session  processes that are well-typed 
according to the Curry-Howard correspondence of linear logic propositions as session types~\cite{CairesP10,CairesPT12,DBLP:conf/icfp/Wadler12}.
This suffices, because 
the type system derived from such a correspondence   ensures 
communication safety, 
session fidelity, and 
deadlock freedom.

\item \fullKoba contains all session  processes that 
enjoy communication safety and session fidelity (as ensured by the type system of Vasconcelos~\cite{V12})
and 
are (dead)lock-free 
by combining Kobayashi's type system based on \emph{usages}~\cite{K02,K06} with Dardha et al.'s encodability result~\cite{DGS12}.
\end{enumerate}

\noindent 
There are good reasons for considering \lcp and \fullKoba.
On the one hand, due to its deep logical foundations,  \lcp appears to us as the \emph{canonic} class of deadlock-free session processes, upon which all other classes should be compared. Indeed, this class arguably offers the most principled yardstick for comparisons.
On the other hand,   \fullKoba integrates session type checking with the sophisticated usage discipline developed by Kobayashi
for  $\pi$-calculus processes.
This indirect approach to deadlock freedom (first suggested in~\cite{K07},
later developed in~\cite{DGS12,DBLP:journals/corr/Dardha14,CDM14})
is fairly general, as it may capture sessions with subtyping, polymorphism, and higher-order communication.
Also, as informally shown in~\cite{CDM14}, \fullKoba strictly includes classes of typed processes induced by other type systems for deadlock freedom in sessions~\cite{DLY07,CD10,P13}.

One key observation in our development is that \fullKoba  corresponds to a \emph{family} of classes of deadlock-free processes, 
denoted $\Koba{0}, \Koba{1}, \cdots, \Koba{n}$,  
which is defined by the \emph{degree of sharing} between their parallel components. 
Intuitively, \Koba{0} is the subclass of \fullKoba 
 with \emph{independent parallel composition}:
for all processes $P \pp Q \in \Koba{0}$, subprocesses $P$ and $Q$ do not share any sessions. Then,  \Koba{1} is the subclass of \fullKoba which contains \Koba{0} but admits also processes with parallel components that share at most one session. Then, \Koba{n} contains deadlock-free session processes whose parallel components share at most $n$ sessions. \\

\noindent\textbf{Contributions.}
In this paper, we present three main contributions:
\begin{enumerate}[1.]
\item 
We show that the inclusion between the constituent classes of \fullKoba is \emph{strict} (Theorem \ref{t:strict}). We have:  
\begin{equation}
\Koba{0} \subset \Koba{1}  \subset \Koba{2} \subset \cdots \subset \Koba{n} \subset \Koba{n+1} \label{eq:hier}
\end{equation}
Although not extremely surprising, the significance of this result lies in the fact that it talks
about concurrency (via the degree of sharing) but implicitly also about the potential sequentiality
of parallel processes. As such, processes in \Koba{k} are necessarily ``more parallel'' than those in \Koba{k+1}.
Interestingly, the degree of sharing in $\Koba{0},\ldots, \Koba{n}$ can be defined in a very simple way,
via a natural condition in the rule for parallel composition in Kobayashi's type system for deadlock freedom.

\item 
We show that  \lcp and \Koba{1} coincide (Theorem \ref{t:cppkoba}). 
That is, there are deadlock-free session processes that cannot be typed by  systems derived from the 
Curry-Howard interpretation of session types~\cite{CairesP10,CairesPT12,DBLP:conf/icfp/Wadler12},
but that can be admitted by the (indirect) approach of~\cite{DGS12}. 
This result is significant: it establishes the precise status of 
systems based on~\cite{CairesPT12,DBLP:conf/icfp/Wadler12} with respect to previous (non Curry-Howard) disciplines. 
Indeed, it formally confirms   that linear logic interpretations of session types naturally induce the most basic form of concurrent cooperation (sharing of exactly one session), embodied as 
the principle of ``composition plus hiding'', a distinguishing feature of such interpretations.

\item 
We define
a rewriting procedure of processes in \fullKoba into \lcp (Defintion \ref{def:typed_enc}).
Intuitively, due to our previous observation and characterization of the degree of sharing in session typed processes,
it is quite natural to convert a process in \fullKoba into another, more parallel process in \lcp. 
In essence, the procedure replaces sequential prefixes with representative parallel components.
The rewriting procedure satisfies type-preservation, and enjoys the compositionality and operational correspondence criteria as stated in~\cite{Gorla10} (cf. Theorems \ref{thm:L0-L2} and \ref{thm:oc}).
These properties not only witness the significance of the rewriting procedure; they also confirm that the degree of sharing is 
a rather subtle criteria for formally distinguishing deadlock-free, session typed processes.
\end{enumerate}

\noindent
To the best of our knowledge, these contributions define the first formal comparison between
fundamentally distinct type systems for deadlock freedom in session communications.
Previous comparisons, such as the ones in~\cite{CDM14} and~\cite[\S6]{CairesPT12}, are informal: they are based on 
representative ``corner cases'', i.e., examples of deadlock-free session processes typable in
one system but not in some other.


The paper is structured as follows.
  \S\,\ref{s:sessions} summarizes the 
    session $\pi$-calculus and associated type system of~\cite{V12}.
In \S\,\ref{s:deadlf} we present the two typed approaches to deadlock freedom for sessions.
 \S\,\ref{s:hier}  defines the classes  \lcp and \fullKoba, 
formalizes the hierarchy~\eqref{eq:hier}, 
and shows that \lcp and \lpkoba coincide.
In \S\,\ref{s:enco} we give the rewriting procedure of \lkoba into \lcp and establish its properties.
\S\,\ref{s:concl} collects some concluding remarks.
Due to space restrictions, details of 
proofs are omitted; they 
can be found online~\cite{DardhaPerez15}.


\section{Session \p calculus}\label{s:sessions}
Following Vasconcelos~\cite{V12},
we introduce the session $\pi$-calculus and its associated type system which ensures communication safety and session fidelity.
The syntax is given in  Figure~\ref{fig:sessionpi} (upper part).
Let $P, Q$ range over processes $x, y$ over channels and $v$ over values; 
for simplicity, the set of values coincides with that of channels.
In examples, we often use $\unit$ to denote a terminated channel that cannot be further used.

Process $\ov x\out v.P$ denotes the output of  $v$ along  $x$, with continuation $P$.
Dually, process $x\inp y.P$ denotes an input along  $x$ with continuation $P$, with 
$y$ denoting a placeholder.
Process $\selection x {l_j}.P$ uses   $x$ to select $l_j$ from a labelled choice process, being
$\branching xlP$, so as to trigger $P_j$; labels indexed by the finite set $I$ are pairwise distinct.
We also have the 
	inactive process (denoted $\nil$),
	the parallel composition of $P$ and $Q$ (denoted $P \pp Q$), 
	and the (double) restriction operator, noted $\res{xy}P$:
	the intention is that $x$ and $y$ denote \emph{dual session endpoints} in $P$.
	We omit $\nil$ whenever possible and write, e.g., $\ov x\out \unit$
	instead of $\ov x\out \unit.\nil$.
	Notions of bound/free variables in processes are standard; we write $\fn{P}$
	to denote the set of free names of $P$.
	Also, we write $P\substj{v}{z}$ to denote the 
	(capture-avoiding) substitution of free occurrences of $z$ in $P$ with $v$.

The operational semantics is given in terms of a reduction relation, noted $P\to Q$, and
defined by the rules in 
 Figure~\ref{fig:sessionpi} (lower part).
It relies on a standard notion of 
 structural congruence, noted $\equiv$~(see~\cite{V12}).
We write $\mto$ to denote the reflexive, transitive closure of $\to$.
 Observe that interaction 
involves prefixes with different channels (endpoints), and 
always occurs in the context of an outermost (double) restriction. 
Key rules are~\Did{R-Com} and \Did{R-Case}, denoting
the interaction of output/input prefixes and  selection/branching constructs, respectively.
Rules~\Did{R-Par}, \Did{R-Res}, and~\Did{R-Str}
are standard.

The syntax of {session types}, ranged over $T, S, \ldots$, is given by the following grammar.
$$
T,S ::= \nilT \midd \wn T.S \midd \oc T.S \midd \branch lS \midd \select lS
$$
Above, $\nilT$ is the type of an endpoint with a terminated protocol.
The type $\wn T.S$ is assigned to an endpoint that first receives a value of type $T$ and then continues according to the protocol described by~$S$.
Dually,  
type $\oc T.S$ is assigned to an endpoint that first outputs a value of type $T$ and then continues according to the protocol described by $S$.
Type
$\select lS$, an \emph{internal choice},  generalizes   output types; type
$\branch lS$, an \emph{external choice},  generalizes input types.
Notice that session types 
describe \emph{sequences} of structured behaviors; they 
do not admit parallel composition operators.

A central notion in session-based concurrency is  \emph{duality}, 
which relates session types offering opposite (i.e., complementary) behaviors.
Duality stands at the basis of communication safety and session fidelity.
Given a session type $T$, its dual type $\dual{T}$ is defined as follows:
  \begin{displaymath}
\begin{array}{rclcrclcrcl}
	\dual{\oc T.S}				&\defeq &  \wn T.\dual{S} &\qquad& \dual{\wn T.S}				&\defeq&  \oc T.\dual{S} & \qquad& & & \\
	\dual{\select lS}			&\defeq &  \branch l{\dual S}  & \qquad & \dual{\branch lS}			&\defeq  & \select l{\dual S} & \qquad&\dual{\nilT}  &\defeq  & \nilT
  \end{array}
\end{displaymath}
\begin{figure}[t]
  \begin{displaymath}
   \begin{array}[t]{rllllll}
	P,Q ::= 	& \ov x\out v.P					& \mbox{(output)} &\quad
	          	& \nil	     						& \mbox{(inaction)} &\\ 
           		& x\inp y.P					& \mbox{(input)}     &\ 
           		& P \pp Q  	    		  		& \mbox{(composition)} &  \\
           		& \selection x {l_j}.P 				& \mbox{(selection)}  & \
			& \res {xy}P					& \mbox{(session restriction)}&\\
         		& \branching xlP				& \mbox{(branching)} &&& \\

	v ::=		& x 							& \mbox{(channel)}&
  \end{array}
  \end{displaymath}
  \begin{displaymath}
    \begin{array}{rllcrll}  
	\Did{R-Com} & 
		\res {xy}(\ov x\out v.P\pp y\inp z.Q) 
		\to 
		\res {xy}(P\pp Q\substj{v}{z})
		& &
		\Did{R-Par}&		{P\to Q}\Longrightarrow
							{P\pp R\to Q\pp R}
      \\[1mm]      
	\!\!\!\!\Did{R-Case} &
		\res {xy}(\selection x{l_j}.P \pp \branching ylP) 
		\to 
		\res {xy}(P\pp P_j)
	~~  j\in I \!\!\!\!\!
	&  &
	\Did{R-Res}&	\!\!\!\!{P\to Q} \Longrightarrow
							{\res {xy} P\to \res {xy} Q}
       \\[1mm]
	\Did{R-Str}& 	{P\equiv P',\ P\to Q,\ Q'\equiv Q}
								\Longrightarrow
								{P'\to Q'} & & &\vspace{2mm}\\
    \end{array}
\end{displaymath}
\vspace{-5mm}
  \caption{Session $\pi$-calculus: syntax and semantics.}
  \label{fig:sessionpi}
\end{figure}
%
\begin{figure}[t]
$$\begin{array}{c}
%
\inferrule[\Did{T-Nil}]
{ }
{x:\nilT\s \nil}
\qquad
\inferrule[\Did{T-Par}]
{\Gamma_1\s P\quad
\Gamma_2\s Q}
{\Gamma_1\circ\Gamma_2\s P\pp Q}
\qquad
\inferrule[{\Did{T-Res}}]
	{
	    \Gamma, x:T, y:\ov{T} \s P
	}
	{ 
	     \Gamma \s \res{xy} P
	}
\qquad
%
\inferrule[{\Did{T-In}}]
	{
 	  \Gamma, x:S, y:T \s P
	}
	{
	  \Gamma, x : \wn T.S \s x\inp{y}.P
	}
\\\\
\inferrule[{\Did{T-Out}}]
	{
	   \Gamma, x:S \s P
	}
	{
	   \Gamma,  x : \oc T.S, y:T \s \ov x\out{y}.P
	}
\qquad
\inferrule[{\Did{T-Brch}}]
	{
		\Gamma, x: S_i \s P_i \quad  \forall i\in I
	}
	{
		\Gamma, x:\branch lS \s \branching {x}lP
	}
\qquad
\inferrule[{\Did{T-Sel}}]
	{
		\Gamma,  x: S_j \s P \quad   \exists j \in I
	}
	{ 	
		\Gamma, x:  \select lS \s \selection {x} {l_j}.P
	}
  \end{array}$$
  \vspace{-5mm}
\caption{Typing rules for the \p calculus with sessions.}
\label{fig:sess_typing}
\end{figure}

Typing contexts, ranged over by $\Gamma, \Gamma'$, are sets of typing assignments $x:T$.
Given a context $\Gamma$ and a process $P$, a session typing judgement is of the form $\Gamma \s P$.
Typing rules are given in Figure~\ref{fig:sess_typing}.
Rule \Did{T-Nil} states that $\nil$ is well-typed under a terminated channel.
Rule \Did{T-Par} types the parallel composition of two processes by composing their corresponding typing contexts
using a splitting operator, noted $\circ$~\cite{V12}.
Rule \Did{T-Res} types a restricted process by requiring that the two endpoints have dual types.
Rules \Did{T-In} and \Did{T-Out} type the receiving and sending of a value over a channel $x$, respectively.
Finally, rules \Did{T-Brch} and \Did{T-Sel} are generalizations of input and output over a labelled set of processes.

The main guarantees of the type system are \emph{communication safety} and \emph{session fidelity}, i.e., typed processes respect their ascribed protocols, as
represented by session types.
\begin{theorem}[Type Preservation for Session Types]\label{thm:subj_red}
If $\Gamma\s P$ and $P\to Q$, then $\Gamma\s Q$.
\end{theorem}

The following notion of well-formed processes is key to single out meaningful typed processes.
\begin{definition}[Well-Formedness for Sessions]\label{def:well-form_sessions.}
A process is {\em well-formed} if for any of its structural congruent processes of the form
$\res{\wt{xy}}(P\pp Q)$ the following hold.
\begin{enumerate}[$\bullet$]
\item
If $P$ and $Q$ are prefixed at the same variable, then the variable performs the 
same action (input or output, branching or selection).

\item
If $P$ is prefixed in $x_i$ and $Q$ is prefixed in $y_i$ where $x_iy_i\in {\wt{xy}}$, then
$P\pp Q\to$.
\end{enumerate}
\end{definition}
It is important to notice that well-typedness of a process does not
imply the process is well-formed. We have the following theorem: 
\begin{theorem}[Type Safety for Sessions~\cite{V12}]\label{thm:safety}
If $~\s P$ then $P$ is well-formed.
\end{theorem}

We present the main result of the session type system.
The following theorem states that
a well-typed closed process does not reduce to an ill-formed one.
It follows immediately from
Theorems~\ref{thm:subj_red} and~\ref{thm:safety}.
\begin{theorem}[\cite{V12}]\label{thm:main}
If $\ \s P$ and $P\mto Q$, then $Q$ is well-formed.
\end{theorem}
An important observation is that the session type system given above does not exclude 
\emph{deadlocked processes}, i.e., 
processes which reach a ``stuck state.''
This is because the interleaving of communication prefixes in typed processes may create extra causal dependencies 
not described by session types. 
(This intuitive definition of deadlocked processes will be made precise below.)
A particularly insidious class of deadlocks is due to cyclic interleaving of channels in processes.
For example, consider a process such as 
%
$
P\defeq \res{xy}\res{wz}(\ov x\out \unit.\ov w\out \unit\pp z\inp t.y\inp s)
$: it
represents the implementation of two (simple) independent sessions, which get intertwined (blocked) due to the 
nesting induced by input and output prefixes.
We have that 
$\unit:\nilT\s P$ even if $P$ is unable to reduce.
A deadlock-free variant of $P$ would be, e.g., process
$P'\defeq \res{xy}\res{wz}(\ov x\out \unit.\ov w\out \unit \pp y\inp s.z\inp t)$, which also is typable in $\s$.

We will say that a  process is \emph{deadlock-free} if any
communication action that becomes active during execution is
eventually consumed; that is, there is a corresponding co-action that eventually becomes available.
Below we  define deadlock freedom 
in the session $\pi$-calculus; 
we follow~\cite{K02,K06} and consider \emph{fair} reduction sequences~\cite{DBLP:journals/iandc/CostaS87}.
For simplicity, we omit the symmetric cases for input and branching.
%
%
\begin{definition}[Deadlock Freedom for Session \p Calculus]\label{def:lock}
  A process $P_0$ is \emph{deadlock-free} if for any fair reduction
  sequence $P_0\to P_1\to P_2\to \ldots$, we have that 
\begin{enumerate}

\item $P_i \equiv \res{\widetilde{xy}}(\ov x\out v.Q \pp R )$, for
  $i\geq 0$, implies that there exists $n \geq i$ such that \\
  $P_n \equiv \res{\widetilde{x'y'}}(\ov x\out v.Q \pp y\inp z.R_1
  \pp R_2)$ and $P_{n+1} \equiv \res{\widetilde{x'y'}}(Q \pp R_1\substj vz
  \pp R_2)$;
%

\item $P_i \equiv \res{\widetilde{xy}}( \selection x {l_j}.Q \pp R)$,
  for $i\geq 0$, implies that there exists $n \geq i$
  such that \\ $P_n\equiv
  \res{\widetilde{x'y'}}(\selection x {l_j}.Q \pp
  y\triangleright\{\mathnormal{l}_k:R_k\}_{k\in I\cup\{j\}}\pp S)$ and
  $P_{n+1} \equiv \res{\widetilde{x'y'}}(Q \pp R_j \pp S)$.
%
\end{enumerate}
\end{definition}

\section{Two Approaches to Deadlock Freedom}\label{s:deadlf}
We introduce two approaches to deadlock-free, session typed processes.
The first one, given in \S\,\ref{ss:lf}, comes from  interpretations of 
 linear logic propositions as 
session types~\cite{CairesCFest14,CairesP10,CairesPT12,DBLP:conf/icfp/Wadler12};  
the second approach,
summarized in \S\,\ref{ss:dgs},
 combines  usage types for the standard \p calculus  with 
 encodings of session processes and types~\cite{DGS12}.
 Based on these two approaches,
 in~\S\,\ref{s:hier} we will define the classes \lcp and \fullKoba.

\subsection{Linear Logic Foundations of Session Types}\label{ss:lf}
The linear logic interpretation of session types was introduced by Caires and Pfenning~\cite{CairesPT12}, and developed
by Wadler~\cite{DBLP:conf/icfp/Wadler12} and others. Initially proposed for intutitionistic linear logic, here we consider 
an interpretation 
based on classical linear logic with mix principles, following a recent presentation by Caires~\cite{CairesCFest14}.

The syntax and semantics of processes are as in \S\,\ref{s:sessions} except for the following differences.
First, we have the standard restriction construct $\res{x}P$, which replaces the 
double restriction.
Second, we have 
 a so-called \emph{forwarding process}, denoted $\linkr{x}{y}$, which intuitively ``fuses'' names $x$ and $y$.
Besides these differences in syntax, we have also some minor modifications in reduction rules. 
Differences with respect to the language considered in \S\,\ref{s:sessions} are summarized in the following:
\begin{displaymath}
   \begin{array}[t]{rllllll}
	P,Q ::= 	
			& \res x P 					& \mbox{(channel restriction)}&\ |\
 			& \linkr{x}{y}					& \mbox{(forwarding)}
  \end{array}
  \end{displaymath}
  \begin{displaymath}
    \begin{array}{rllcrll}  
	\Did{R-ChCom} & 
		\ov x\out v.P\pp x\inp z.Q
		\to 
		P\pp Q\substj{v}{z}
		& \quad & 
		\Did{R-Fwd} & 
		\res {x}(\linkr{x}{y} \pp P)
		\to 
		P\substj{y}{x}
      \\
	\Did{R-ChCase} &
		\selection x{l_j}.P \pp \branching xlP
		\to 
		P\pp P_j
	\quad  j\in I 
       & \quad &  
	\Did{R-ChRes}&	{P\to Q} \Longrightarrow
							{\res {x} P\to \res {x} Q}
    \end{array}
\end{displaymath}
Observe how interaction of input/output prefixes and selection/branching is no longer covered by an outermost restriction.
As for the type system, we consider the so-called C-types
which correspond to linear logic propositions. They
 are given by the following grammar:
$$
A,B ::= \bot \midd \one \midd A\otimes B  \midd A \parl B \midd  \select lA \midd \branch lA  
$$
Intuitively, $\bot$ and $\one$ are used to type a terminated endpoint.
Type $A\otimes B$ is associated to an endpoint that first outputs an object of type $A$ and then 
behaves according to $B$. Dually, type 
$A \parl B$ is the type of an endpoint that first inputs an object of type $A$ and then continues as $B$.
The interpretation of $\select lA$ and $\branch lA$ as select and branch behaviors follows as expected.

We define a full duality on C-types, which exactly
corresponds to the negation operator of \CLL\ $(\cdot)^\perp$.
The \emph{dual} of type $A$, denoted $\dual{A}$, is inductively defined
as follows:
\begin{displaymath}
\begin{array}{rclcrclcrcl}
\dual{\one}&=&\bot &\qquad& \dual{\bot}&=&\one & \quad & \dual{\select lA} & = & \branch l{\dual A} \\
\dual{A\otimes B} & = & \dual{A} \parl \dual{B} &\qquad& \dual{A\parl B} & = & \dual{A} \otimes \dual{B} &\qquad& \dual{\branch lA} & = &\select l{\dual A}
\end{array}
\end{displaymath}
Recall that $A \lolli B \triangleq \dual{A} \parl B$.
As explained in~\cite{CairesCFest14},
considering mix principles means admitting $\bot \lolli \one$ and
$\one \lolli \bot$, and therefore $\bot = \one$.
We write $\bullet$ to denote either $\bot$ or $\one$,
and decree that $\dual{\bullet} = \bullet$. 

Typing contexts, sets of typing assignments $x:A$, are ranged over $\Delta, \Delta', \ldots$.
The empty context is denoted `$\,\cdot\,$'.
Typing judgments are then of the form $P \cp \D$.
Figure~\ref{fig:type-system-cll} gives the typing rules associated to the linear logic interpretation.
Salient points include the use of bound output $\res{y}\ov{x}\out y.P$,
which is abbreviated as $\bout{x}{y}P$.
Another highlight is the 
``composition plus hiding'' principle implemented by rule $\Did{T-$\cut$}$, 
which integrates parallel composition and restriction in a single rule.
Indeed, there is no dedicated rule for restriction.
Also, rule \Did{T-$\mix$} enables the typing of \emph{independent parallel compositions}, i.e.,
the composition of two processes that do not share sessions.
%
%
\begin{figure*}[t]
\centering
{
$$
\begin{array}{c}
\inferrule*[left=\Did{T-$\one$}]{
}{ \nil \cp x{:}\bullet}
\qquad
\inferrule*[left=\Did{T-$\bot$}]
{P \cp  \Delta}
{P \cp x{:}\bullet, \Delta}
\qquad
\inferrule*[left=\Did{T-$\tid$}]{
}
{\linkr{x}{y} \cp x{:}A, y{:}\dual{A}}
\vspace{2mm}\\
%
\inferrule[\Did{T-$\parl$}]
{P \cp \Delta, y{:}A, x{:}B}
{x\inp {y}.P \cp \Delta, x{:} A\parl B}
\qquad
\inferrule[\Did{T-$\otimes$}]
{P \cp  \Delta, y{:}{ A}  \and Q \cp \Delta', x{:}B}
{\bout{x}{y}. (P\para Q) \cp  \Delta, \Delta', x{:}  A\otimes B}
\qquad
\inferrule[\Did{T-$\cut$}]
{P \cp \Delta, x{:}\dual{A}  \and Q \cp  \Delta', x{:}A}
{\res {x}(P\para Q) \cp\Delta, \Delta'}
\vspace{2mm}\\
\inferrule[\Did{T-$\oplus$}]
{P \cp \Delta, x{:}A_j  \and j\in I}
{ \selection {x} {l_j}.P  \cp \Delta, x{:}\select lA}
\qquad
\inferrule[\Did{T-$\with$}]
{P_i \cp \Delta, x{:}A_i  \and \forall i\in I}
{\branching xlP \cp \Delta, x{:}\branch lA}
\qquad
\inferrule[\Did{T-$\mix$}]
{P \cp \Delta  \and Q \cp  \Delta'}
{ P\para Q \cp\Delta, \Delta'}
\end{array}
$$
}
\vspace{-5mm}
\caption{\label{fig:type-system-cll}
Typing rules for the \p calculus with C-types.}
\end{figure*}

We now collect main results for this type system; see~\cite{CairesPT12,CairesCFest14} for details.
For any $P$, define $live(P)$ if and only if $P \equiv (\nub \wt{n})(\pi.Q \para R)$, 
where $\pi$ is an input, output, selection, or branching prefix.

\begin{theorem}[Type Preservation for C-Types]
If $P \cp \D$ and $P \redd Q$ then $Q \cp \D$.
\end{theorem}

\begin{theorem}[Progress]\label{t:progress}
If $ P \cp \cdot $ and $live(P)$ then $P \redd Q$,  for some $Q$.
\end{theorem}

\subsection{Deadlock Freedom by Encodability}\label{ss:dgs}
As mentioned above,
the second approach to deadlock-free session processes is \emph{indirect},
in the sense that  
establishing deadlock freedom for session processes
appeals to usage types for the $\pi$-calculus~\cite{K02,K06}, 
for which type systems enforcing deadlock freedom are well-established.
Formally, this reduction 
exploits encodings of processes and types:
a session process 
$\Gamma \s P$ 
is encoded into a (standard) $\pi$-calculus process 
$\encf \Gamma \lf \encf P$.
Next we introduce 
the syntax of standard $\pi$-calculus processes with variant values
(\S\,\ref{sss:proc}), 
the discipline of usage types (\S\,\ref{sss:usage}), 
and 
the encodings of session processes and types into standard $\pi$-calculus processes and usage types, respectively (\S\,\ref{sss:enco}).

\subsubsection{Processes}\label{sss:proc}
The syntax and semantics of the \p calculus with usage types build upon those in \S\,\ref{s:sessions}.
We require some modifications.
First, the encoding of terms presented in \S\,\ref{sss:enco}, requires polyadic communication.
Rather than branching and selection constructs, the \p calculus that we consider here includes a \emph{case} construct
$\picase v{x_i}{P_i}$ that uses \emph{variant value} $\vv {j}v$.
Moreover, we consider the standard channel restriction, rather than  double restriction.
These modifications are summarized below: 
\begin{displaymath}
   \begin{array}[t]{rllllll}
	P,Q ::= 	
			& \res x P 					& \mbox{(channel restriction)}&\ |\
 			& \picase v{x_i}{P_i}				& \mbox{(case)}&\\
	v::=
			&\vv jv						&\mbox{(variant value)}
 			
  \end{array}
  \end{displaymath}
  \begin{displaymath}
 \begin{array}{lllll}  
	\Did{R\p Com} & 
		\ov x\out {\wt v}.P\pp x\inp {\wt z}.Q
		\to 
		P\pp Q\substj{\wt v}{\wt z}
	\\
		\Did{R\p Res}&	{P\to Q} \Longrightarrow
							{\res {x} P\to \res {x} Q}
      \\
	\Did{R\p Case} &
		 \picase {\vv jv}{x_i}{P_i}	
		\to 
		P_j\substj{v}{x_i}
	\quad  j\in I 
    \end{array}
\end{displaymath}
The definition of deadlock-freedom for the $\pi$-calculus follows~\cite{K02,K06}:

\begin{definition}[Deadlock Freedom for Standard $\pi$-Calculus]\label{def:lockpi}
  A process $P_0$ is \emph{deadlock-free} under fair scheduling,
  if for any fair reduction
  sequence $P_0\to P_1\to P_2\to \cdots$ the following hold 
\begin{enumerate}

\item if $P_i \equiv \res{\tilde{x}}(\ov x\out {\wt v}.Q \pp R )$ for
  $i\geq 0$, implies that there exists $n \geq i$ such that \\
  $P_n\equiv \res{\tilde{x}}(\ov x\out {\wt v}.Q \pp x\inp {\wt z}.R_1
  \pp R_2)$ and $P_{n+1} \equiv \res{\tilde{x}}(Q \pp R_1\substj{{\wt v}}{{\wt z}}
  \pp R_2)$;

 \item if $P_i \equiv \res{\tilde{x}}( x \inp {\wt z}.Q \pp R )$ for
    $i\geq 0$, implies that there exists $n \geq i$ such that \\ $P_n\equiv
   \res{\tilde{x}}(x \inp {\wt z}.Q \pp \ov x\out {\wt v}.R_1 \pp R_2)$ and
   $P_{n+1} \equiv \res{\tilde{x}}(Q\substj{{\wt v}}{{\wt z}} \pp R_1 \pp R_2)$.
\end{enumerate}
\end{definition}

\subsubsection{Usage Types}\label{sss:usage}
\begin{figure}[t]
  \begin{displaymath}
  \begin{array}[t]{rllllll}
	U ::=
				& \wn^{\ob}_{\ca}.U  		& \mbox{(used in input)}&\quad
                              	& {\emp}  					&  \mbox{(not usable)}\\ 
                              	& \oc^{\ob}_{\ca}.U  			&  \mbox{(used in output)}&
                              	& (U_1 \pp U_2) 			&  \mbox{(used in parallel)}\\
    	T ::=
    				& U [\widetilde T]			& \mbox{(channel types)}& 
    				& \variant {l}{T}				& \mbox{(variant type)}
\end{array}
  \end{displaymath}
  \vspace{-5mm}
  \caption{Syntax of usage types for the \p calculus.}
  \label{fig:usage_types}
\end{figure}
The syntax of usage types is defined in Figure~\ref{fig:usage_types}.
For simplicity, we let $\alpha$ range over input $\wn$ or output $\oc$ actions.
The usage $\emp$ describes a channel that cannot be used at all. 
We will often omit $\emp$, and so we will write $U$ instead of $U.\emp$.
Usages $ \wn^{\ob}_{\ca}.U$ and $\oc^{\ob}_{\ca}.U$
describe channels that can be used once for input and output, respectively and then used according to the continuation usage $U$. 
The \emph{obligation} $o$ and \emph{capability} $\ca$ range over the set of natural numbers.
The usage $U_1\pp U_2$ describes a channel that is used according to $U_1$ by one process
and $U_2$ by another processes in parallel.

Intuitively, obligations and capabilities describe inter-channel dependencies:
\begin{enumerate}[$\bullet$]
\item
An obligation of level $n$ must be fulfilled by using only capabilities of level \emph{less than} $n$.
Said differently, an action of obligation $n$ must be prefixed by actions of capabilities less than $n$.

\item
For an action with capability of level $n$, there must exist a co-action with obligation of level \emph{less than or equal to} $n$.
\end{enumerate}

\noindent
Typing contexts are sets of typing assignments and  are ranged over $\Gamma, \Gamma'$.
A typing judgement is of the form
$\Gamma \lf P$: the annotation $n$ explicitly denotes the greatest \emph{degree of sharing} admitted in parallel processes.
Before commenting on the typing rules (given in Figure~\ref{f:pityping}), we discuss some important auxiliary notions, extracted 
from~\cite{K02,K06}. 
First, the composition operation on types (denoted $\pp$, and 
used in rules  T$\pi$-\Did{Par}$_n$ and  T$\pi$-\Did{Out})
is based on the composition of usages and is defined as follows:
\begin{displaymath}
\begin{array}{c}
\variant {l_i}{T_i} \pp \variant {l_i}{T_i}  = \variant {l_i}{T_i}\qquad
U_1[\wt T] \pp U_2[\wt T]					 = (U_1\pp U_2)[\wt T]
\end{array}
\end{displaymath}
The generalization of $\pp$ to typing contexts, denoted $(\Gamma_1 \pp \Gamma_2)(x)$, is defined as expected.
%
%
The unary operation $\uparrow^{\,t}$ 
applied to a usage $U$ lifts its
obligation level {\em up to t}; it is defined inductively as:
\begin{align*}
\uparrow^{\,t} \emp 					= \emp \qquad 
\uparrow^{\,t} \m{\alpha}^{\ob}_{\ca}.U = \m{\alpha}^{max(\ob,t)}_{\ca}.U\qquad
\uparrow^{\,t} (U_1\pp U_2)			&= (\uparrow^{t} U_1\pp \uparrow^{t} U_2)
\end{align*}
The $\uparrow^{\,t}$ is extends to types/typing contexts as expected.
%
%
\emph{Duality} on usage types 
simply exchanges $\wn$ and $\oc$:
\begin{displaymath}
	\begin{array}{c}
		\dual{{\emptyset} []}			=				{\emptyset}[]\qquad
		\dual{\wn^{\ob}_{\ca}.U [\wt T]}	=				 \oc^{\ob}_{\ca}.\dual U[\wt T]\qquad
		\dual{\oc^{\ob}_{\ca}.U [\wt T]}	=				 {\wn^{\ob}_{\ca}.\dual U [\wt T]}
	\end{array}
\end{displaymath}
Operator ``$\semi$'' in $\Delta = x:[T]\alpha^{\ob}_{\ca} \semi \Gamma$,
used in rules \Did{T$\pi$-In} and \Did{T$\pi$-Out},
 is such that the following hold:
\begin{align*}
   \dom(\Delta) = \{x\} \cup \dom(\Gamma) \qquad\quad
   \Delta(x)=\left\{ 
    \begin{aligned}
      &\alpha^{\ob}_{\ca}.U[\widetilde T]		&& \mbox{if 
$\Gamma(x)=U[\widetilde T]$}\\
      &\alpha^{\ob}_{\ca}[\widetilde T]		&& \mbox{if $x\notin 
dom(\Gamma)$}
    \end{aligned}
  \right. \qquad\quad
   \Delta(y) = \uparrow^{\ca+1}\Gamma(y) \ \ \mbox{if $y\neq x$}
\end{align*}
The final required notion is that of a \emph{reliable usage}. It builds upon the following definition:
\begin{definition}
\label{def:obligations}
Let $U$ be a usage.
The input and output \emph{obligation levels} (resp. \emph{capability levels}) of $U$, written
$\obs{\wn}{U}$ and $\obs{\oc}{U}$
(resp. $\caps{\wn}{U}$ and $\caps{\oc}{U}$),
are defined as:
$$
\begin{array}{rclcrcl}
\obs\alpha{\alpha^{\ob}_{\ca}.U}					&= &  \ob & \quad & \caps\alpha{\alpha^{\ob}_{\ca}.U}					&= &\ca \\
\obs\alpha{U_1\pp U_2}							&= & min(\obs\alpha{U_1},\obs\alpha{U_2})
& \quad& 
\caps\alpha{U_1\pp U_2}							&= & min(\caps\alpha{U_1},\caps\alpha{U_2})
\end{array}
$$
\end{definition}
%
%
\noindent
The definition of reliable usages depends on a reduction relation on usages, noted $U\to U'$.
Intuitively, $U\to U'$
 means that if a channel of usage $U$ is used for communication, then
after the communication occurs, the channel should be used according to usage $U'$.
Thus, e.g.,  $\wn^{\ob}_{\ca}.U_1 \pp \wn^{\ob'}_{\ca'}.U_2$
		reduces to 
		$U_1\pp U_2$.
		
\begin{definition}[Reliability]\label{d:reli}
We write $\conpar{\alpha}{U}$ when $\obs{\ov\alpha}{U} \leq \caps{\alpha}{U}$.
We write $\con{U}$ when  $\conpar\wn U$ and $\conpar\oc U$ hold.
Usage $U$ is \emph{reliable}, noted $\rel U$, if $\con {U'}$ holds 
$\forall U'$ such that $U\mto U'$.
\end{definition}

\paragraph{Typing Rules.}
The typing rules for the standard \p calculus with usage types are given in Figure~\ref{f:pityping}.
The only difference with respect to the rules in Kobayashi's  systems~\cite{K02,K06} is that we annotate typing judgements with the degree of sharing, explicitly stated in rule \Did{T$\pi$-Par$_n$}--see below.
Rule~\Did{T\p Nil} states that the terminated process is typed under a terminated channel.
Rule~\Did{T\p Res} states that  process $\res xP$ is well-typed if the usage for $x$
 is reliable  (cf. Definition~\ref{d:reli}).
Rules \Did{T\p In} and \Did{T\p Out} type input and output processes in a typing context where the ``$\semi$'' operator is used in order to
increase the obligation level of the channels   in  continuation $P$.
Rules \Did{T\p LVal} and \Did{T\p Case} type a choice: the first types a variant value with a variant type; the second types a case process using a variant value as its guard.

Given a degree of sharing $n$, rule~\Did{T$\pi$-Par$_n$} 
states that the parallel composition of processes $P$ 
and $Q$ 
(typable under contexts $\Gamma_1$ and $\Gamma_2$, respectively)
is well-typed under the  typing context $\Gamma_1 \pp \Gamma_2$ only if 
$|{\Gamma_1} \cap {\Gamma_2} | \leq n$.
This allows to simply characterize the ``concurrent cooperation'' between $P$ and $Q$.
As a consequence, 
if $P \vdash^{n}_{\mathtt{KB}}$ then $P \vdash^{k}_{\mathtt{KB}}$, for any $k \leq n$.
Observe that the typing rule for parallel composition in~\cite{K02,K06} is the same as \Did{T$\pi$-Par$_n$}, except for  
condition $|{\Gamma_1} \cap {\Gamma_2} | \leq n$, which is not specified.

The next theorems imply that well-typed processes by the type system in Figure~\ref{f:pityping} are deadlock-free.
\begin{theorem}[Type Preservation for Usage Types]
If $\Gamma \lf P$ and $P\to Q$, then
$\Gamma'\lf Q$ for some $\Gamma'$ such that $\Gamma\to \Gamma'$.
\end{theorem}

\begin{theorem}[Deadlock Freedom]\label{t:dfk}
If $\emptyset \lf P$ and
either $P\equiv \res{\tilde{x}}(x\inp {\wt z}.Q \pp R )$ or $P\equiv \res{\tilde{x}}(\ov x\out {\wt v}.Q \pp R )$,
then $P\to Q$, for some $Q$.
\end{theorem}

\begin{corollary}
If $\emptyset \lf P$, then $P$ is deadlock-free, in the sense of Definition~\ref{def:lockpi}.
\end{corollary}

{Theorem~\ref{t:progress} (progress for the linear logic system) 
and Theorem~\ref{t:dfk} (deadlock freedom for the standard \p calculus) have a rather  similar formulation:   both properties 
state that processes 
can always reduce if they 
are well-typed (under the empty typing context) and have an appropriate structure (i.e., 
condition 
$live(P)$ in Theorem~\ref{t:progress} and 
condition 
$P\equiv \res{\tilde{x}}(x\inp {\wt z}.Q \pp R )$ or $P\equiv \res{\tilde{x}}(\ov x\out {\wt v}.Q \pp R )$
in Theorem~\ref{t:dfk}).}

\subsubsection{Encodings of Processes and Types}\label{sss:enco}

\begin{figure*}[t]
\centering
$$\begin{array}{c}
\inferrule[{\Did{T$\pi$-Nil}}]{
}{
	x: \emp [] \lf {\nil}
}
      \qquad
      %
    %
\inferrule[{\Did{T$\pi$-Res}}]{
      \Gamma, x: U[\widetilde T] \lf P\quad \rel{U}
}{
      \Gamma \lf \res{x} P
}
\qquad
%
\inferrule[{\Did{T$\pi$-Par$_n$}}]{
	\Gamma_1 \lf P\and \Gamma_2 \lf Q
    \\\\  |{\Gamma_1} \cap {\Gamma_1} | \leq n
}{ 
      \Gamma_1 \pp \Gamma_2 \lf  P\pp Q
}
\qquad
\inferrule[{\Did{T$\pi$-In}}]{
      \Gamma, \widetilde y:\widetilde T\lf P
}{
      x:\wn^{0}_{\ca}[\widetilde T]\semi \Gamma \lf x\inp{\tilde y}.P
}
\\\\
\inferrule[{\Did{T$\pi$-Out}}]{
      \Gamma_1 \lf \tl v:  \widetilde T \and \Gamma_2 \lf P
}{
       x: \oc^{0}_{\ca}[\widetilde T] \semi (\Gamma_1\pp \Gamma_2) \lf {\ov x}\out{ \tl v}.P
}
\qquad
\inferrule[{\Did{T$\pi$-LVal}}]{
      \Gamma \lf v: T_j \quad \exists j\in I
}{
      \Gamma \lf   {l_j}\_v: \variant {l_i}{T_i}
}
    \qquad
\inferrule[{\Did{T$\pi$-Case}}]{
      \Gamma_1 \lf v:  \variant {l_i}{T_i} \\\\ 
      \Gamma_2, x_i:T_i \lf P_i \quad \forall i\in I
}{
      \Gamma_1 , \Gamma_2 \lf  \picase v{x_i}{P_i}
}
  \end{array}$$
  \vspace{-5mm}
  \caption{Typing rules for the \p calculus with usage types with degree of sharing $n$.}\label{f:pityping}
\end{figure*}

\paragraph{Encoding of Processes.}
To relate classes of processes obtained by the different type systems given so far,
we rewrite a session typed or C-typed process into a usage typed process
by following a continuation-passing style:
this allows us
to mimic the structure of a session  or C-type by sending its continuation as a payload over a channel.
This idea, suggested in~\cite{K07} and developed in~\cite{DGS12}, is recalled in Figure~\ref{d:encdgs}.

\begin{figure}
\begin{displaymath}
\begin{array}{llcl}
	\encf{\ov x\out v.P}				&\defeq \ \res c \ov{\f x}\out {v,{c}}.\encx P{c}
\\
	\encf{x\inp y.P}					&\defeq \ \f x\inp{y,c}.\encx P{c}
\\
 	 \encf{\selection x {l_j}.P}			&\defeq \  \res c \ov{\f x}\out {\vv{j}c}.\encx Pc
\\
  \end{array}
  \qquad
\begin{array}{llcl}
	 \encf{\branching xlP}			&\defeq \ \f x \inp y.\ \picase {y}{c}{\encx {P_i}c}
\\
	 \encf{\res{xy}P}				&\defeq \ \res {c} \enc P{c}
\\
	 \encf{P\pp Q}					&\defeq \ \encf{P} \pp \encf{Q}
\\
  \end{array}
\end{displaymath}
\vspace{-5mm}
\caption{Encoding of session processes into $\pi$-calculus processes.}
\label{d:encdgs}
\end{figure}

\paragraph{Encoding of Types.}
We 
formally relate 
session types and
logic propositions
to
usage types by means of the encodings
given in Figure~\ref{f:enctypes}.
The former one, denoted as denoted $\enco{\cdot}\su$, is taken from~\cite{DGS12}.

\begin{figure}[t]
\centering
\begin{minipage}{.45\textwidth}
\begin{eqnarray*}
\enco{\nilT}{\su} & = & \emp [] \\
\enco{\wn T.S}{\su} & = & \inpuse{\ob}{\ca} [\enco {T}\su, \enco S\su]\\
\enco{\oc T.S}{\su} & = &\outuse{\ob}{\ca} [\enco{T}{\su} ,\enco{\dual S}{\su}]\\
\enco{\branch lS}{\su} & = & \inpuse \ob\ca [\variant {l_i}{\enco {S_i}\su}]\\
\enco{\select lS}{\su} & = &  \outuse \ob\ca [\variant {l_i}{\enco{\dual{S_i}}\su}]
\end{eqnarray*}
\end{minipage}
\begin{minipage}{.45\textwidth}
\begin{eqnarray*}
\enco{\nilT}{\ct} & = & \bullet \\
\enco{\wn T.S}{\ct} & = &  \enco{{T}}{\ct} \parl \enco{S}{\ct}\\
\enco{\oc T.S}{\ct} & = & \enco{\dual T}{\ct} \otimes \enco{S}{\ct}\\
\enco{\branch lS}{\ct} & = & \&\{l_i:\enco{S_i}\ct\}_{i\in I} \\
\enco{\select lS}{\ct} & = &  \oplus\{l_i:\enco{S_i}\ct\}_{i\in I} 
\end{eqnarray*}
\end{minipage}
\caption{Encodings of session types into usage types (Left) and C-types (Right).
\label{f:enctypes}}
\end{figure}

\begin{definition}
\label{def:enc_env_su}
\label{def:enc_env_sc}
Let $\Gamma$ be a session typing context. The encoding $\encf{\cdot}$ into
usage typing context  and $\enco{\cdot}{\ct}$ into C-typing context is inductively defined as follows:
\[
\encf{\emp}=\enco{\emp}{\ct} \defeq \emp
\qquad 
\encf{\Gamma,x:T}\defeq \encf{\Gamma} , \f x:\enco{T}{\su}
\qquad
\enco{\Gamma,x:T}{\ct}	\defeq \enco{\Gamma}{\ct}, x:\enco{T}{\ct}
\]
\end{definition}
%
\begin{lemma}[Duality and encoding of session types]\label{lem:dualenc}
Let $T,S$ be finite session types. \\ Then:
(i)~ $\dual T =  S$  if and only if $\dual{\enco{T}{\ct}} = \enco{S}{\ct}$;
(ii)~$\dual T =  S$  if and only if $\dual{\enco{T}{\su}} = \enco{S}{\su}$.
%
%
\end{lemma}
%

\paragraph{On Deadlock Freedom by Encoding.}
The next results relate deadlock freedom, typing and encoding.
\begin{proposition}
Let $P$ be a deadlock-free session process, then
$\encf P$ is a deadlock-free \p process.
\end{proposition}
\begin{proof}
Follows
by the encoding of terms given in Figure~\ref{d:encdgs}, Definition~\ref{def:lock} and Definition~\ref{def:lockpi}.
\end{proof}
Next we recall an important result relating deadlock freedom and typing, by following~\cite{CDM14}.
\begin{corollary}
Let $\s P$ be a session process.
If $\lf \encf P$ is deadlock-free then $P$ is deadlock-free.
\end{corollary}


\section{A Hierarchy of Deadlock-Free Session Typed Processes}\label{s:hier}
\paragraph{Preliminaries.}
To formally define the classes   \lcp and \fullKoba, we require some auxiliary definitions.
The following translation addresses 
minor syntactic differences between session typed processes (cf. \S\,\ref{s:sessions}) 
and the processes typable in the linear logic interpretation of session types (cf. \S\,\ref{ss:lf}).
Such differences concern output actions and the restriction operator:

\begin{definition}\label{d:trans}
Let $P$ be a session process. The translation $\chr{\cdot}$ is defined as
$$
\begin{array}{rclcrcl}
\chr{\overline x\out y.P} & = & \bout{x}{z}.(\linkr{z}{y} \pp \chr{P}) & \qquad & 
\chr{\res{xy}P} & = & \res{w}\chr{P}\substj{w}{x}\substj{w}{y} \quad w \not\in \fn{P}
\end{array}
$$
and as an homomorphism for the other process constructs.
\end{definition}

\noindent 
Let $\enco{\cdot}{\ct}$ denote the encoding of session types
into linear logic propositions in Figure~\ref{f:enctypes} (right).
Recall  that $\encf{\cdot}$ stands for the encoding of processes and $\enco{\cdot}{\su}$ for the encoding of types,
both defined in~\cite{DGS12},
and given here in Figure~\ref{d:encdgs} and Figure~\ref{f:enctypes} (left), respectively.
We may then formally define the languages under comparison as follows:

\begin{definition}[Typed Languages]\label{d:lang}
The languages \lcp and \Koba{n}($n \geq 0$) are defined as follows:
\begin{eqnarray*}
\lcp  & = &   \big\{P \mid \exists \Gamma.\ (\Gamma \s P \,\land\, \chr{P} \cp \enco{\Gamma}{\ct}) \big\} \\
\lkoba & = & \big\{  P \mid \exists \Gamma,f.\ (\Gamma \s P \,\land\, \encf \Gamma \lf \encf P) \big\} 
\end{eqnarray*}
\end{definition}

\paragraph{Main Results.} 
Our first observation is that there are processes in $\Koba{2}$ but not in $\Koba{1}$:
\begin{lemma}
\label{lem:K1inK2}
 $\lpkoba \subset \Koba{2}$.
 \end{lemma}
 \begin{proof}
 $\Koba{2}$ contains (deadlock-free) session processes not captured in $\lpkoba$.
 A representative example is: 
$$
P_2 =    \res{a_1b_1}\res{a_2b_2}(a_1\inp x.\;\ov{a_2}\out{x}
        \pp \ov{b_1}\out{\unit}.\;b_2\inp{z})
$$
%
\noindent
This process is not in $\lpkoba$ because it
involves the composition of two parallel processes which share two sessions. 
As such, it is typable in $\lf$ (with $n \geq 2$) but not in $\plf$.
 \end{proof}
The previous result generalizes easily, so as to define a hierarchy of deadlock-free, session processes:
\begin{theorem}\label{t:strict}
For all $n \geq 1$, we have that $\Koba{n} \subset \Koba{n+1}$.
\end{theorem}

\begin{proof}
Immediate by considering one of the following processes, which generalize process $P_2$ in Lemma~\ref{lem:K1inK2}:
\begin{eqnarray*}
P_{n+1} & = &     \res{a_1b_1}\res{a_2b_2}\cdots \res{a_{n+1}b_{n+1}} (a_1\inp x.\;\ov{a_2}\out{x}.\cdots.\;\ov{a_{n+1}}\out{y}
        \pp \ov{b_1}\out{\unit}.\;b_2\inp{z}.\;\cdots\;b_{n+1}\inp{z}) \\
Q_{n+1} & = &     \res{a_1b_1}\res{a_2b_2}\cdots \res{a_{n+1}b_{n+1}} (a_1\inp x.\;\ov{a_2}\out{x}.\cdots.\;a_{n+1}\inp {y}
        \pp \ov{b_1}\out{\unit}.\;b_2\inp{z}.\;\cdots\;\ov{b_{n+1}}\out{\unit})
\end{eqnarray*}
To distinguish $\Koba{n+1}$ from $\Koba{n}$, we consider $P_{n+1}$ if $n+1$ is even and $Q_{n+1}$ otherwise.
\end{proof}

One main result of this paper is that $\lcp$ and $\lpkoba$ coincide.
Before stating this result, 
we make the following observations.
The typing rules for processes in \lcp do not directly allow free output.
However, free output is representable (and typable) by linear logic types by means of the transformation
in Definition~\ref{d:trans}.
Thus, considered processes are not syntactically equal. 
In \lcp there is
cooperating composition (enabled by rule \Did{T-$\cut$} in Figure~\ref{fig:type-system-cll}); independent composition can only be enabled by rule \Did{T-$\mix$}.
Arbitrary restriction is not allowed; only restriction of parallel  processes.

The following property is key in our developments: it connects our encodings of (dual) session types into usage types
with  reliability (Definition~\ref{d:reli}), a central notion to the type system for deadlock freedom in Figure~\ref{f:pityping}.
Recall that, unlike usage types, there is no parallel composition operator at the level of session types.
\begin{proposition}\label{p:relia}
Let $T$ be a session type. 
Then $\rel{\enco{T}{\su} \pp \enco{\dual{T}}{\su}}$ holds.
\end{proposition}

\begin{proof}[Proof (Sketch)]
By induction on the structure of session type $T$ and
the definitions of $\enco{\cdot}{\su}$
and predicate $\rel{\cdot}$, using 
Lemma~\ref{lem:dualenc} (encodings of types preserve session type duality).
See \cite{DardhaPerez15} for details. \end{proof}

%
%

We then have the following main result, whose  proof is detailed 
in~\cite{DardhaPerez15}:
\begin{theorem}\label{t:cppkoba}
$\lcp = \lpkoba$.
\end{theorem}

%
%
%
%
%
%
%
%
%

Therefore, we have the following corollary, which attests that 
the class of deadlock-free session processes naturally induced by 
linear logic interpretations of session types is strictly included in the class
induced by the indirect approach of Dardha et al.~\cite{DGS12} (cf. \S\,\ref{ss:dgs}).

\begin{corollary}
$\lcp \subset \lkoba$,  $n >1$.
\end{corollary}
\noindent 
The fact that (deadlock-free) processes such as $P_2$ (cf. Lemma~\ref{lem:K1inK2}) are not in \lcp is informally discussed in~\cite[\S6]{CairesPT12}. However, \cite{CairesPT12} gives no formal comparisons with other classes of deadlock-free processes.

%

\section{Rewriting \lkoba into \lcp}\label{s:enco}
The hierarchy of deadlock-free session processes established by Theorem \ref{t:strict} is \emph{subtle} in the following sense:
if $P \in \Koba{k+1}$ but $P \not\in \Koba{k}$ (with $k \geq 1$) then we know that there is a subprocess of $P$ that needs to be ``adjusted'' in order to ``fit in'' $\Koba{k}$. More precisely, we know that such a subprocess of $P$ must become more independent in order to be typable under the lesser degree of sharing $k$. 

Here we propose a \emph{rewriting procedure} that converts processes in $\Koba{n}$ into processes in $\Koba{1}$ (that is, $\lcp$, by Theorem \ref{t:cppkoba}). 
The rewriting procedure follows a simple idea: given a parallel process as input, return as output a process in 
which 
one of the components is kept unchanged, but the other is replaced by parallel representatives of the sessions implemented in it.
Such parallel representatives are formally defined as characteristic processes and catalyzers, introduced next.
The rewriting procedure is type preserving and satisfies operational correspondence 
(cf. Theorems \ref{thm:L0-L2} and \ref{thm:oc}).

\subsection{Preliminaries: Characteristic Processes and Catalyzers}
Before presenting our rewriting procedure, let us first introduce some preliminary results.
\begin{definition}[Characteristic Processes of a Session Type]\label{def:charprocess}
Let $T$ be a
session type (cf. \S\,\ref{s:sessions}).
Given a name $x$, 
the set of \emph{characteristic processes} of $T$, denoted \chrp{T}{x}, is inductively defined as follows:
$$
\begin{array}{rcl}
\chrp{\nilT}{x} 			& = & \big\{ P \mid P \cp x{:}\bullet \big\} \\
\chrp{\wn T.S}{x} 		& = &  \big\{x(y).P \mid P \cp y{:}\enco{{T}}{\ct}, x{:} \enco{S}{\ct}\big\}\\
\chrp{\oc T.S}{x} 		& = & \big\{\bout{x}{y}.(P \para Q) \mid P \in \chrp{\dual{T}}{y} \land Q \in \chrp{S}{x}\big\}\\
\chrp{\branch lS}{x}				& = &  \big\{\branching xlP \mid \forall i\in I.\ P_i\in \chrp{S_i}x\big\}\\
\chrp{\select lS}x 				& = &  \bigcup_{i\in I}\big\{ \selection x{l _i}.{P_i} \mid P_i\in \chrp{S_i}x\big\}
\end{array}
$$
\end{definition}

\begin{definition}[Catalyzer]\label{def:catalyser}
Given a session typing context $\Gamma$, we define its associated \emph{catalyzer} as a process context $\mathcal{C}_\Gamma[\cdot]$, as follows:
\begin{eqnarray*}
\mathcal{C}_\emptyset[\cdot]  =  [\cdot] \qquad \qquad 
\mathcal{C}_{\Gamma, x:T}[\cdot]  =  (\nub x)(\mathcal{C}_{\Gamma}[\cdot] \para P) \quad \text{with $P \in \chrp{\dual{T}}{x}$}
\end{eqnarray*}
\end{definition}
\noindent
We record the fact that characteristic processes are well-typed in the system of \S\,\ref{ss:lf}:

\begin{lemma}
\label{lem: wt_characteristic}
Let $T$ be a session type.
For all $P \in \chrp{T}{x}$, we have:  $P \cp x:\enco{T}{\ct}$
\end{lemma}
We use $\chrp{T}{x} \cp x:\enco{T}{\ct}$ to denote the set of processes $P \in \chrp{T}{x}$ such that $P \cp x:\enco{T}{\ct}$.
%
%
\begin{lemma}[Catalyzers Preserve Typability]
\label{lem:typability_catal}
Let $\Gamma\s P$ and $\Gamma'\subseteq \Gamma$.
Then
$\catal{\Gamma'}{P}\cp  \enco{\Gamma}{\ct} \setminus \enco{\Gamma'}{\ct}$.
\end{lemma}
%

\begin{corollary}
Let $\Gamma\s P$. Then $\catal{\Gamma}{P}\cp\emptyset$.
\end{corollary}

\subsection{Rewriting \lkoba in \lcp}
We start this section with some notations.
First,
in order to represent  pseudo-non deterministic binary choices between two equally typed processes,
we introduce the following:

\begin{notation}\label{con:fakepar}
Let $P_1$, $P_2$ be two processes such that $k \not\in\fn{P_1,P_2}$.
We write $P_1 \fakepar{k} P_2$ to stand for the process
$
\res{k}(\selection k\inx.\nil \pp {k\triangleright\{\inl :P_1, \inr: P_2 \}})
$,
where label $\inx$ stands for either $\inl$ or $\inr$.
\end{notation}
Clearly, 
since session execution is purely deterministic, 
notation $P_1 \fakepar{k} P_2$  denotes  that
either $P_1$ or $P_2$ will be executed (and that the actual deterministic choice is not relevant).
It is worth adding that Caires has already developed the technical machinery required to include non deterministic behavior 
into the linear logic interpretation of session types; see~\cite{CairesCFest14}.
Casting our rewriting procedure into the typed framework of~\cite{CairesCFest14},
so as to consider actual non deterministic choices,
  is interesting future work.

We find it convenient to 
annotate bound names in processes with session types, and write $\res {xy:T}P$
and $x\inp {y:T}.P$, for some session type $T$.
When the reduction relation involves a left or right choice in a binary labelled choice, 
as in reductions due to pseudo-non deterministic choices (Notation~\ref{con:fakepar}),
we sometimes annotate the reduction as $\to^\inl$ or $\to^\inr$.
We let $\mathtt C$ denote a \emph{process context}, i.e., a process with a hole.
And finally, for a typing context $\Gamma$,
we shall write $\chrp{\Gamma}{}$
to denote the process
$
\prod_{(w_i:T_i)\in \Gamma}\,\chrp{T_i}{w_i}
$.
We are now ready to give the rewriting procedure from \lkoba to \lcp.
\begin{definition}[Rewriting $\lkoba$ into $\lcp$]\label{def:typed_enc}
Let $P \in \lkoba$
such that $\Gamma\s P$, for some $\Gamma$.
The encoding $\encoCP{\Gamma\s P}{}$ is a process of $\lcp$ inductively defined as follows:
\begin{align*}
\encCP{{x:\nilT} \s \nil}					& \defeq  \nil
&& \\[1mm]
\encCP{{\Gamma} \s\dual{x}\out v.P'}		& \defeq   \dual{x}(z). \big(\linkr{v}{z}\para \encCP{{{\Gamma'}, x:{S}}\s P'} \big)
&&\Gamma = \Gamma', x:\oc T.S, v:T\\[1mm]
\encCP{{\Gamma}\s x\inp {y:T}.P'}			& \defeq  x(y).\encCP{{\Gamma',x:S,y:T}\s P'}										&&\Gamma = \Gamma', x:\wn T.S	\\[1mm]
\encCP{{\Gamma} \s\selection x{l_j}.P'}		& \defeq  \selection x{l_j}.{\encCP{{\Gamma',x:S_j}\s P'}}										&&\Gamma = \Gamma', x:\select lS\\[1mm]  			
\encCP{{\Gamma}\s \branching xlP}		&  \defeq  \parbranching x{l_i}{\encCP{{\Gamma',x:S_i}\s P_i}}
&&\Gamma = \Gamma' ,x:\branch lS \\[1mm]
\encCP{{\Gamma}\s \res {\wt{xy}:\wt{S}}(P\para Q)}& \defeq  \chrp{\Gamma_2}{} \para \catal{\wt z:\wt S}{\encCP{{\Gamma_1, {\wt x{:}\wt S}}\s P}\substj{\wt z}{\wt x}} 
&& \Gamma= \Gamma_1\circ\Gamma_2\ \wedge \Gamma_1, \wt {x}:\wt {S} \s P \\[1mm]									
& \quad ~~ \fakepar{k} 
\chrp{\Gamma_1}{} \para \catal{\wt z:\wt {V}}{\encCP{{\Gamma_2,{\wt y{:}\wt {V}}}\s Q}\substj{\wt z}{\wt y}} 
&&  \Gamma_2, \wt {y}:\wt {V}\s Q \wedge V_i = \dual{S_i}
\end{align*}
\end{definition}
We illustrate the procedure 
in~\cite{DardhaPerez15}.
Notice that the rewriting procedure given in Definition~\ref{def:typed_enc}
satisfies the compositionality criteria given in~\cite{Gorla10}. In particular, it is easy to see that 
the rewriting of a composition of terms is defined in terms of the rewriting  of the constituent subterms.
Indeed, e.g.,  
$\encCP{\Gamma_1 \circ \Gamma_2 \s \res {xy:S}(P\para Q)}$ depends on a context including
both $\encCP{\Gamma_1, x:S   \s P}$
and 
$\encCP{\Gamma_2, y:\dual{S}   \s Q}$.

We present two important results about our rewriting procedure.
First, we show it is type 
preserving:

\begin{theorem}[Rewriting is Type Preserving]
\label{thm:L0-L2}
Let $(\Gamma\s P) \in \lkoba$. Then, $\encCP{{\Gamma}\s P}\cp \enco{\Gamma}{\ct}$.
\end{theorem}

\noindent
Notice that the inverse of the previous theorem is trivial by following the definition of typed encoding.
{Theorem~\ref{thm:L0-L2} is meaningful, for it says that the type interface of a process (i.e., the set of sessions implemented in it) is not modified by the rewriting procedure. 
That is, the procedure modifies the process structure by closely following the 
causality relations described by (session) types.
Notice that 
causality relations present in processes, but not described at the level of types,
may be removed.

The rewriting procedure also satisfies an operational correspondence result. 
Let us write $\Gamma \s P_1, P_2$ whenever both $\Gamma \s P_1$ and
$\Gamma \s P_2$ hold. We have the following auxiliary definition:
}

%



\begin{definition}
Let $P, P'$ be such that $\Gamma\s P,P'$.
Then, we write $P\uptok P'$ if and only if
$P= \mathtt C[Q]$ and $P' = \mathtt C[Q']$, for some context $\mathtt C$, and there is $\Gamma'$ such that
$\Gamma'\s Q, Q'$.
\end{definition}

\begin{theorem}[Operational Correspondence]\label{thm:oc}
Let $P \in \lkoba$ 
such that $\Gamma\s P$ for some $\Gamma$. Then we have:
\begin{enumerate}[I)]
\item
If $P\to P'$ then there exist $Q$, $Q'$ s.t.
(i) $\encCP{{\Gamma}\s P} \to^{\inx}\mto\equiv Q$;
(ii) $Q\uptok  Q'$; 
(iii) $\encCP{{\Gamma}\s P'}\to^{\inx} Q'$.
\item
If $\encCP{{\Gamma}\s P}  \to^{\inx}\mto\equiv Q$ then there exists $P'$ s.t.
$P\to P'$ and $Q\uptok  \encCP{\Gamma \s P'}$.
%
\end{enumerate}
\end{theorem}

%


\section{Concluding Remarks}\label{s:concl}

%

We have presented a formal comparison of fundamentally distinct type systems for deadlock-free, 
session typed processes. To the best of our knowledge, ours is the first work to establish precise relationships of this kind. 
Indeed, prior comparisons 
between type systems for deadlock freedom
are informal, given in terms of representative 
examples typable in one type system but not in some other. 

An immediate difficulty in giving a unified account of different typed frameworks for deadlock freedom is the
variety of 
process languages, type structures, and typing rules that define each framework. 
Indeed, our comparisons involve: 
the framework of session processes put forward by Vasconcelos~\cite{V12}; 
the interpretation of linear logic propositions as session types by Caires~\cite{CairesCFest14};
the $\pi$-calculus with usage types defined by Kobayashi in~\cite{K02}.
Finding some common ground for comparing these three frameworks is not trivial---several translations/transformations
were required in our developments to account for numerous syntactic differences.
We made an effort to follow the exact definitions in each framework.
Overall, we believe that we managed to 
concentrate on essential semantic features of two 
salient classes of deadlock-free session processes, noted \lcp and \fullKoba.

Our main contribution is identifying the \emph{degree of sharing} as
a subtle, important issue that underlies both session typing and deadlock freedom.
We propose a simple 
characterization
of the degree of sharing: in essence,
it arises via an explicit premise for the typing rule for parallel composition in the type system in~\cite{K02}.
The degree of sharing is shown to effectively induce a strict hierarchy of deadlock-free session processes in \fullKoba, as
resulting from the approach of~\cite{DGS12}.
We showed that the most elementary (and non trivial) member of this hierarchy precisely corresponds 
to \lcp--arguably the most canonical class of session typed processes known to date. 
Furthermore, by exhibiting an intuitive rewriting procedure of processes in $\fullKoba$ into
processes in $\lcp$, we demonstrated that the degree of sharing is a
subtle criteria for distinguishing deadlock-free processes.
As such, even if our technical developments are  technically simple, in our view they 
substantially clarify our  understanding of 
type systems for  
liveness properties (such as deadlock freedom) 
in the context of $\pi$-calculus processes.

As future work, we would like to obtain \emph{semantic characterizations} of the degree of sharing, in the form of, e.g., 
  preorders on typed processes that distinguish when one process ``is more parallel'' than another.
We plan also to extend our formal relationships to cover typing disciplines with \emph{infinite behavior}.
We notice that the approach of~\cite{DGS12} extends to recursive behavior~\cite{DBLP:journals/corr/Dardha14} and that 
infinite (yet non divergent) behavior has been incorporated 
into logic-based session types~\cite{DBLP:conf/tgc/ToninhoCP14}.
Finally, we plan to explore whether the 
rewriting procedure given in \S\,\ref{s:enco} could be adapted into a  \emph{deadlock resolution} procedure.
%
\paragraph{Acknowledgements.}
We are grateful to Lu\'{i}s Caires, Simon J. Gay, and the anonymous reviewers for their 
valuable 
comments 
and suggestions.
This work was partially supported by the EU COST Action IC1201 ({Behavioural Types for Reliable Large-Scale Software Systems}).
 Dardha is supported by the UK EPSRC project
EP/K034413/1
 ({From Data Types to Session Types: A Basis  for Concurrency and Distribution}).
 P\'{e}rez is  also affiliated to NOVA  Laboratory for Computer Science and Informatics, Universidade Nova de Lisboa, Portugal.
 
\bibliographystyle{eptcs}
\bibliography{my_biblio}


\end{document}